%% file: mse_paper_final.tex
\documentclass[10pt]{book}
\usepackage{enumitem}
\usepackage{array}
\usepackage{comment}
\usepackage{enumitem}
\usepackage[tbtags]{amsmath}
\usepackage{amsthm}
\usepackage{amssymb}
\usepackage{amsmath}
\usepackage{amsfonts}
\usepackage{epsfig}
\usepackage{graphicx}
\usepackage[small]{caption}
\usepackage[sort&compress]{natbib}
\bibpunct{(}{)}{;}{a}{}{,}
\usepackage{longtable}

\newcommand{\commentt}[1]{}

\newcommand{\that}        {\mbox{$\hat{\boldsymbol{\theta}}$}}

\newcommand{\thhj}{\hat{\theta}_j}

\newcommand{\V}{\text{Var}}

\newcommand{\tr}{\text{tr}}

\newcommand{\bhat}{\boldsymbol{\hat{\beta}}}

\newcommand{\pxv}{{P}_X^V}

\newcommand{\bu}{\boldsymbol{u}}

\newcommand{\uw}{\boldsymbol{w}}
\newcommand{\bz}{\boldsymbol{z}}

\newcommand{\bzero}{\boldsymbol{0}}

\newcommand{\util}{\boldsymbol{\tilde{u}}}

\newcommand{\btil}{\boldsymbol{\tilde{\beta}}}

\newcommand{\bxj}{\boldsymbol{x_{j}}}
\newcommand{\bei}{\boldsymbol{e_{i}}}
\newcommand{\bej}{\boldsymbol{e_{j}}}

\newcommand{\sig}        {\Sigma}

\newcommand{\thi}        {\theta_i}
\newcommand{\thhi}        {\mbox{$\hat{\theta}_i$}}

\newcommand{\thbi}        {\mbox{$\hat{\theta}_i^B$}}
\newcommand{\thbis}        {\mbox{$\hat{\theta}_{i*}^B$}}

\newcommand{\thbj}        {\mbox{$\hat{\theta}_j^B$}}
\newcommand{\thbjs}        {\mbox{$\hat{\theta}_{j*}^B$}}

\newcommand{\thebi}       {\mbox{$\hat{\theta}_i^{EB}$}}
\newcommand{\thebj}       {\mbox{$\hat{\theta}_j^{EB}$}}
\newcommand{\theblupi}    {\mbox{$\hat{\theta}_i^{EBM1}$}}
\newcommand{\thbmi}    {\mbox{$\hat{\theta}_i^{BM1}$}}

\newcommand{\thhw}       {\mbox{$\bar{\hat{\theta}}_{w}$}}

\newcommand{\thbarw}     {\mbox{$\bar{\hat{\theta}}_{w}^{B}$}}
\newcommand{\thbarwb}     {\mbox{$\bar{\hat{\theta}}_w^{B}$}}

\newcommand{\thbarweb}     {\mbox{$\bar{\hat{\theta}}_w^{EB}$}}

\newcommand{\su}    {\sigma_u^2}

\newcommand{\sut}     {\mbox{$\tilde{\sigma}_u^2$}}
\newcommand{\suh}     {\mbox{$\hat{\sigma}_u^2$}}
\newcommand{\sus}     {\mbox{${\sigma}_u^{*2}$}}
\newcommand{\sust}     {\mbox{$\tilde{\sigma}_u^{*2}$}}

\newcommand{\bb}     {\mbox{${\boldsymbol{\beta}}$}}

\newcommand{\lp}     {\left(}
\newcommand{\rp}     {\right)}
\newcommand{\lb}     {\left[}
\newcommand{\rb}     {\right]}

\newcommand{\xg}     {\mbox{$(X'V^{-1}X)^{-1}X'$}}
\newcommand{\bxi}{\boldsymbol{x}_i}

\newcommand{\ma}     {\max_{1 \leq i \leq m}}

\newcommand{\cov}     {\text{Cov}}

\newtheorem{theorem}{Theorem}

\newtheoremstyle{lemma}
{\topsep} % space above
{\topsep} % space below
{\it} % body font
{} % indent
{\bf} % head font in lower caps bolded
{:} % punctuation between head and body
{0.5em} % space after head
{} % manually specify head
\theoremstyle{lemma}
\newtheorem{lem}{Lemma}

\newtheoremstyle{example}
{\topsep} % space above
{\topsep} % space below
{} % body font
{} % indent
{\bf} % head font in lower caps bolded
{:} % punctuation between head and body
{0.5em} % space after head
{} % manually specify head
\theoremstyle{example}

\newtheoremstyle{remark}
{\topsep} % space above
{\topsep} % space below
{} % body font
{} % indent
{} % head font in lower caps bolded
{:} % punctuation between head and body
{0.5em} % space after head
{} % manually specify head
\theoremstyle{remark}
\newtheorem{rem}{Remark}

\newcommand{\sustt}     {\mbox{${\hat{\sigma}}_u^{*2}$}}
\newcommand{\thebs}     {\hat{\theta}_i^{EB*}}

\begin{document}
\renewcommand{\baselinestretch}{1.2}
%\lhead[\fancyplain{} \leftmark]{}
%\chead[]{}
%\rhead[]{\fancyplain{}\rightmark}
%\cfoot{}
%\headrulewidth=0pt
\markright{
%\hbox{\footnotesize\rm Statistica Sinica
%{\footnotesize\bf ??}(200?), 000-000}\hfill
}
\markboth{\hfill{\footnotesize\rm REBECCA C. STEORTS AND MALAY GHOSH
}\hfill}
{\hfill {\footnotesize\rm ESTIMATION OF MEAN SQUARED ERRORS} \hfill}
\renewcommand{\thefootnote}{}
$\ $\par
\fontsize{10.95}{14pt plus.8pt minus .6pt}\selectfont
\vspace{0.8pc}
\centerline{\large\bf ON ESTIMATION OF MEAN SQUARED ERRORS OF}
\vspace{2pt}
\centerline{\large\bf  BENCHMARKED EMPIRICAL BAYES ESTIMATORS}
\vspace{.4cm}
\centerline{Rebecca C. Steorts and Malay Ghosh}
\vspace{.4cm}
\centerline{\it Carnegie Mellon University and University of Florida}
\vspace{.55cm}
\fontsize{9}{11.5pt plus.8pt minus .6pt}\selectfont

%
% \includecomment{notreport}
% \excludecomment{report}

\excludecomment{notreport}
\includecomment{report}

% \fontsize{14}{16pt}\selectfont
% \fontsize{13}{15pt}\selectfont

%\maketitle
%\thispagestyle{empty}
%\newpage
%\tableofcontents

%\vfil

%\eject %

%\setcounter{page}{1}
%\input{introduction}
%\input{lit}
%\input{onestage}
%\input{twostage}
%\input{summary}
%\input{mse_abstract}

%%%====FILES FOR MSE PAPER========
\input{mse_abstract2} %%DONE check
\input{mse_intro}  %%DONE check
\input{mse_new3}%%DONE check

\input{estimation_NEW}  %check

%%%new stuff on bootstrapping
\input{bootstrap} %check
%%%new stuff - take out for ghosh comments

\input{mse_applic}
\input{mse_summary}
%%don't put in thesis
\input{appendixa}
\input{appendixb}
%\input{supplementary}
%\input{lemma_sup3} %need to format thesis version
%\input{newlemma4_NEW}

%\newpage
%\baselineskip12pt
\bibliography{mse_paper_final}
\bibliographystyle{asa.bst}
%\nocite{battese_1988,datta_2000,datta_2005,datta_1991,ghosh_1994, ghosh_1992, rao_2003, you_2003, you_2004, wang_2008}
\nocite{rao_2003, you_2003, you_2004, wang_2008}

\vskip .65cm
\noindent
Department of Statistics, Carnegie Mellon University, Baker Hall 232K, Pittsburgh, PA 15213, 
U.S.A.
\vskip 2pt
\noindent
E-mail: (beka@cmu.edu)
\vskip 2pt
\noindent
Department of Statistics, University of Florida, P.O. Box 118545, Gainesville, FL 32611-8545, 
U.S.A.
\vskip 2pt
\noindent
E-mail: (ghoshm@stat.ufl.edu)
\vskip .3cm

\end{document}

%% file: mse_abstract2.tex
\begin{quotation}
\noindent {\it Abstract:}

We consider benchmarked empirical Bayes (EB) estimators under the basic area-level model of Fay and Herriot while requiring the standard benchmarking constraint. In this paper we determine the excess mean squared error (MSE) from constraining the estimates through benchmarking. We show that the increase due to benchmarking is $O(m^{-1}),$ where $m$ is the number of small areas.
Furthermore, we find an asymptotically unbiased estimator of this MSE and compare it to the second-order approximation of the 
MSE of the EB estimator or, equivalently, of the MSE of the empirical best linear unbiased predictor
(EBLUP), that was derived by \citet*{prasad_1990}. Morever, using methods similar to those of \citet*{butar_2003}, we compute a parametric bootstrap estimator of the MSE of the benchmarked EB estimator under the Fay-Herriot model and compare it to the MSE of the benchmarked EB estimator found by a second-order approximation. Finally, we illustrate our methods using SAIPE data from the U.S. Census Bureau, and in a simulation study.

\vspace{9pt}
\noindent {\it Key words and phrases:}
Small-area, Fay-Herriot, Mean Squared Error, Empirical Bayes, Benchmarking, Parametric Bootstrap \par
\end{quotation} \par

\fontsize{10.95}{14pt plus.8pt minus .6pt}\selectfont

%% file: mse_intro.tex
%%%%% INTRODUCTION 
\setcounter{chapter}{1}
\setcounter{equation}{0}
\noindent {\bf 1. Introduction}

Small area estimation has become increasingly popular recently due to a growing demand for such statistics. It is well known that direct small-area estimators usually have large standard errors and coefficients of variation. In order to produce estimates for these small areas, it is necessary to borrow strength from other related areas. Accordingly, model-based estimates often differ widely from the direct estimates, especially for areas with small sample sizes. One problem that arises in practice is that the model-based estimates do not aggregate to the more reliable direct survey estimates.
Agreement with the direct estimates is often a political necessity to convince legislators of the utility of small area estimates.
The process of adjusting model-based estimates to correct this problem is known as benchmarking. Another key benefit of benchmarking is protection against model misspecification as pointed out by \citet*{you_2004} and \citet*{datta_2011}. 
%benchmark the model-based estimates so that the benchmarked hierarchical Bayes (HB) estimates add up to the direct estimate at the higher level of aggregation. 

In recent years, the literature on benchmarking has grown. Among others, \citet*{pfeffermann_1991}; \citet*{you_2003}; \citet*{you_2004}; \citet*{pfeffermann_2006}; and \citet*{ugarte_2009} have made an impact on the continuing development of this field.  Specifically, \citet*{wang_2008} provided a frequentist method wherein an augmented model was used to construct a best linear unbiased predictor (BLUP) that automatically satisfies the benchmarking constraint.  In addition, \citet*{datta_2011} developed very general benchmarked Bayes estimators, that covered most of the earlier estimators that were motivated from either a frequentist or Bayesian perspective. Specifically, they found benchmarked Bayes estimators under the \mbox{\citet*{fay_1979}} model. 
%These results were used to benchmark the state-level proportion of school children under poverty to the corresponding national estimate.\\

%
%proposed a more general constrained Bayes procedure where we can benchmark the weighted mean or both the weighted mean and weighted variability for a general loss. \\

Due to the fact that they borrow strength, model-based estimates typically show a substantial improvement over direct estimates in terms of mean squared error (MSE). It is of particular interest to determine how much of this advantage is lost by constraining the estimates through benchmarking. The aforementioned work of \citet*{wang_2008} 
and \citet*{ugarte_2009} examined this question through simulation studies but did not derive any probabilistic results. They showed that the MSE of the benchmarked EB estimator was slightly larger than the MSE of the EB estimator for their simulation studies. In Section 3, we derive a second-order approximation of the MSE of the benchmarked Bayes EB estimator to show that the increase due to benchmarking is $O(m^{-1}),$ where $m$ is the number of small areas.

% In Section 3, we show that the increase due to benchmarking is $O(m^{-1}),$ where $m$ is the number of small areas.

In this paper, we are concerned with the basic area-level model of \citet*{fay_1979}. We propose benchmarked EB estimators in Section 2. In Section 3, we derive a second-order asymptotic expansion of the MSE of the benchmarked EB estimator. In Section~4, we find an estimator of this MSE and compare it to the
second-order approximation of the MSE of the EB estimator or, equivalently, the 
MSE of the EBLUP, that was derived by \mbox{\citet*{prasad_1990}}. 
%We then compare the estimate of our MSE approximation to the estimate of the MSE of the  EBLUP in \mbox{\citet*{prasad_1990}.} 
Finally, in Section 5, using methods similar to those of \mbox{\citet*{butar_2003}}, we compute a parametric bootstrap estimator of the mean squared error of the benchmarked EB estimator under the Fay-Herriot (1979) model and compare it to our estimators from Section 2. Section 6 contains an application based on Small Area Income and Poverty Estimation Data (SAIPE) from the U.S. Census Bureau as well as a simulation study. Some concluding remarks are made in Section 7.\par

%%%%%% SECTION 2
\setcounter{chapter}{2}
\setcounter{equation}{0} %-1
\noindent {\bf 2. Benchmarked Empirical Bayes Estimators}
%\section{Benchmarked Empirical Bayes Estimators}
\label{sect2}

Consider the area-level random effects model
\begin{align}
\label{re}
\thhi = \thi + e_i,\quad \thi = \bxi^T \bb + u_i,\quad i=1,\ldots,m;
\end{align}
where $e_i$ and $u_i$ are mutually independent with 
$e_i \stackrel{ind}{\sim} N(0,D_i)$ and 
$u_i \stackrel{iid}{\sim} N(0,\su).$ This model was first 
considered in the context of estimating income for small areas (population less than 1000) by \citet*{fay_1979}.
In (\ref{re}), the $D_i$ are known as are the $p \times 1$ design vectors $\bxi.$ However, the vector of regression coefficients $\bb_{p \times 1}$ is unknown. 

When the variance component $\su$ is known and $\bb$ has a uniform prior on $\mathbb{R}^p,$ then the Bayes estimator of $\thi$ is given by $\thbi = (1-B_i)\thhi + B_i\bxi^T\btil$ where $B_i = D_i(\su + D_i)^{-1},$ $\btil \equiv  \btil(\su) = \xg V^{-1}\that,$ and $V = \text{Diag}(\su + D_1, \ldots, \su + D_m).$ Suppose now we want to match the weighted average of some estimates $\delta_i$ to the weighted average of the direct estimates, which we denote by~$t.$ We assume for our calculations that 
$t = \sum_i w_i\hat{\theta_i} =: \thhw.$ We denote the normalized weights by $w_i,$ so that $\sum_i w_i = 1.$  Under the loss $L(\theta, \delta) = 
\sum_i w_i (\thi - \delta_i)^2,$ and subject to $\sum_i w_i \delta_i = \sum_i w_i \thhi, $ the benchmarked Bayes estimator derived in \citet*{datta_2011} is 
\begin{align}
\label{bench1}
\thbmi &= \thbi + (\thhw-\thbarw),\quad i=1,\ldots,m;
\end{align}
where $\thbarw = \sum_i w_i \hat{\theta}_i^B.$
In more realistic settings, $\su$ is unknown. Let $P_X = X(X^TX)^{-1}X^T,$ $h_{ij} = \bxi^T(X^TX)^{-1}\bxj,$ $\hat{u}_i = \thhi - \bxi^T\bhat,$ and $\bhat = (X^TX)^{-1}X^T\that.$ In this paper, we consider the simple moment estimator given by $\suh = \max\{0,\sut\}$ where $\sut= (m-p)^{-1}\left[
\sum_{i=1}^m \hat{u}_i^2 - \sum_{i=1}^m D_i (1 - h_{ii})
\right]$, which is given in \citet*{prasad_1990}. Then the  benchmarked EB estimator of $\thi$ is 
\begin{align}
\label{fayest}
\theblupi = \thebi + (\thhw-\thbarweb),
\end{align}
where $\thebi = (1-\hat{B}_i)\thhi + \hat{B}_i\bxi^T \btil(\suh),\; \hat{B}_i = D_i(\suh + D_i)^{-1},\;
i=1,\ldots,m.$ The objective of the next two sections will be to obtain the MSE of the benchmarked EB estimator correct up to $O(m^{-1}) $ and also to find an estimator of the MSE correct to the same order. \par

%% file: mse_new3.tex
\setcounter{chapter}{3}
\setcounter{equation}{0}
\noindent {\bf 3. Second-Order Approximation to MSE}
%\section{Second-Order Approximation to MSE}
%\label{s_mse}

\cite{wang_2008} construct a simulation study to compare the MSE of the benchmarked EB estimator to the MSE of the EB estimator. 
In this section, we derive a second order expansion for the MSE of the benchmarked Bayes estimator under the same regularity conditions and assuming the standard benchmarking constraint. That is, for the model proposed in Section 2, we obtain a second-order approximation to the MSE of the empirical benchmarked Bayes estimator derived in Section~2. Take $h_{ij}^V = \bxi^T(X^TV^{-1}X)^{-1}\bxj$ and assume that $\su > 0.$ 
Establishing Theorem \ref{thm_mse} requires the regularity conditions 
%\begin{itemize}
%\item[(i)]  $0 <D_L \leq \inf_{1 \leq i \leq m}D_i  \leq \sup_{1 \leq i \leq m}D_i \leq D_U < \infty;$
%%$\sup_{i \geq 1}D_i \leq D_U < \infty$ and 
%%$\inf_{i \geq 1}D_i \geq D_L > 0;$
%\item[(ii)] $\max_{1 \leq i \leq m} h_{ii} = O(m^{-1});$ and
%\item[(iii)] $\ma w_i = O(m^{-1}).$
%\end{itemize}
\begin{itemize}[itemsep=-5.5pt]
\item[(i)]  $0 <D_L \leq \inf_{1 \leq i \leq m}D_i  \leq \sup_{1 \leq i \leq m}D_i \leq D_U < \infty;$
%$\sup_{i \geq 1}D_i \leq D_U < \infty$ and 
%$\inf_{i \geq 1}D_i \geq D_L > 0;$
\item[(ii)] $\max_{1 \leq i \leq m} h_{ii} = O(m^{-1});$ and
\item[(iii)] $\ma w_i = O(m^{-1}).$
\end{itemize}

%\begin{enumerate}
%\item  $0 <D_L \leq \inf_{1 \leq i \leq m}D_i  \leq \sup_{1 \leq i \leq m}D_i \leq D_U < \infty;$
%%$\sup_{i \geq 1}D_i \leq D_U < \infty$ and 
%%$\inf_{i \geq 1}D_i \geq D_L > 0;$
%\item $\max_{1 \leq i \leq m} h_{ii} = O(m^{-1});$ and
%\item $\ma w_i = O(m^{-1}).$
%\end{enumerate}

Condition $(iii)$ requires a kind of homogeneity of the small areas, and in particular, it assumes there are not a few large areas that dominate the others in terms of the $w_i.$ Conditions $(i)$ and $(ii)$ are similar to those of \citet*{prasad_1990} and are often assumed in the small area estimation literature.

Before stating Theorem 1, we first present some lemmas whose proofs are provided in the supplementary material and are used in the proof of Theorem ~1. The proof of Theorem~1 can be found in Appendix B.

\begin{lem} Let $r>0$ be arbitrary. Then
%Let $\util = \that - X\btil.$ 
%Consider the following collection of results where $C_1, C_2,$ and $C_3$ depend only on $D_U,$ $D_L,$ and $X$:
\label{lem_comb1}
\begin{itemize}
\item[(i)]
%%i
\label{part_thib}
$E\lb \left\{\dfrac{\partial \thbi}{\partial \su}\right\}^{2r} \rb = O(1),$ and
\item[(ii)]
%%ii
\label{part_thibs_sq}
$E\lb\sup_{\su \geq 0}\left|\dfrac{\partial^2 \thbi}{\partial (\su)^2}\right|^{2r} \rb = O(1).$ 
\end{itemize}
\end{lem}

%%% proof of lem 1
%% i
\begin{notreport}
\begin{proof}[Proof of Lemma \ref{lem_comb1}(i)]
Recall $\util = \that - X\btil$ and define $\bu = \that - X\bb.$
Recall $\thbi = (1-B_i)\thhi +B_i\bxi^T\btil.$ Since 
$\dfrac{\partial \btil}{\partial \su} = -(X^TV^{-1}X)^{-1}X^TV^{-2}\that,$ we can easily show that
$
\dfrac{\partial \thbi}{\partial \su} 
= [B_i^2D_i^{-1}\bei^T - B_i\bxi^T(X^TV^{-1}X)^{-1}X^TV^{-2}]\util.
$
This implies that 
$$\left|\dfrac{\partial \thbi}{\partial \su}\right| 
\leq |D_L^{-1}\bei^T\util| + |\bxi^T(X^TV^{-1}X)^{-1}X^TV^{-2}\util|. $$

Then
\begin{align}
\left|\dfrac{\partial \thbi}{\partial \su}\right|^{2r} \notag
&\leq 2^{2r-1}D_L^{-2r}|\bei^T\util|^{2r}
+ 2^{2r-1}|\bxi^T(X^TV^{-1}X)^{-1}X^TV^{-2}\util|^{2r} \notag \\
&\leq 2^{2r-1}D_L^{-2r}|\thhi - \bxi^T\btil|^{2r}
+ 
2^{2r-1} \lb\bxi^T(X^TV^{-1}X)^{-1}\bxi \util^T V^{-3} \util\rb ^{r}\notag \\
&\leq 
2^{2r-1}D_L^{-2r}|\thhi - \bxi^T\btil|^{2r}
+ 2^{2r-1} \lb (\ma h_i) (\su + D_U)(\su + D_L)^{-1} D_L^{-1} \bu^T V^{-1} \bu\rb ^{r}\notag \\
&\leq 2^{2r-1}D_L^{-2r}|\thhi - \bxi^T\btil|^{2r}
+2^{2r-1} \lb  (1 + D_UD_L^{-1})D_L^{-1}(\ma h_i) \bu^T V^{-1} \bu\rb ^{r} = O_p(1) \notag.
\end{align}
since $\thhi - \bxi^T\btil \sim N(0,V_i),$ $\ma h_i = O(m^{-1}),$ and $\bu^TV^{-1}\bu \sim \chi^2_m.$
This implies that $E\lb\left|\dfrac{\partial \thbi}{\partial \su}\right|^{2r}\rb = O(1).$ 
\end{proof}
\end{notreport}
%
%% ii
\begin{notreport}
\begin{proof}[Proof of Lemma \ref{lem_comb1}(ii)]
Recall $\pxv = X(X^TV^{-1}X)^{-1}X^TV^{-1}, \util = \that - X\btil,$ and \mbox{$\bu = \that - X\bb.$}
Knowing  $\dfrac{\partial \btil}{\partial \su} = -(X^TV^{-1}X)^{-1}X^TV^{-2}\that,$ we find that
\begin{align*} 
\dfrac{\partial^2 \thbi}{\partial (\su)^2} &= 
-2B_i^3D_i^{-2}\bei^T\util + 2B_i^2D_i^{-1}\bei^T\pxv V^{-1}
 \util \\
&+ 2B_i\bei^T \pxv V^{-2} \util - 2 B_i \bei^T \pxv V^{-1}
 \pxv V^{-1} \util.
\end{align*}
Then we find
\begin{align}
\label{sec_der}
\left|\dfrac{\partial^2 \thbi}{\partial (\su)^2} \right|
&\leq 
|2B_i^3D_i^{-2}\bei^T\util| + |2B_i^2D_i^{-1}\bei^T\pxv V^{-1}
 \util|\notag\\
 &+ |2B_i\bei^T \pxv V^{-2} \util| + |2 B_i \bei^T \pxv V^{-1}
 \pxv V^{-1} \util|.
\end{align} 
Define $\Omega = (X^TV^{-1}X)^{-1}(X^TV^{-2}X)(X^TV^{-1}X)^{-1}(X^TV^{-2}X)(X^TV^{-1}X)^{-1}.$ Using
our expression in (\ref{sec_der}), we find that
\begin{align}
\label{first_der}
&\left|\dfrac{\partial^2 \thbi}{\partial (\su)^2} \right|^{2r} \notag \\
& \qquad \leq 4^{3r-1} \bigg[
|B_i^3D_L^{-2}\bei^T\util|^{2r}
+ B_i^{4r}D_L^{-2r} [\bxi^T(X^TV^{-1}X)^{-1}\bxi\util^TV^{-3}\util]^r \notag \\
&\qquad + B_i^{2r}[\bxi^T(X^TV^{-1}X)^{-1}\bxi\util^TV^{-5}\util]^r\notag  + B_i^{2r} [\bxi^T \Omega \bxi\util^TV^{-3}\util]^r \bigg]\notag \\
%%%%%
&\qquad \leq  4^{3r-1}\bigg[
B_i^{r}D_L^{-4r}|\thhi - \bxi^T\bb|^{2r}
+D_L^{-3r} [  (1 + D_U D_L^{-1}) (\ma h_i) \util^TV^{-1}\util]^r\notag \\
&\qquad +
D_L^{-3r}[ (1 + D_U D_L^{-1})  (\ma h_i) \util^TV^{-1}\util]^r]\notag\\
&\qquad+ 
 D_L^{-3r} [(1 + D_U D_L^{-1})(\ma h_i ) \util^TV^{-1}\util]^r
\bigg]\notag \\
%%%
&\qquad =   4^{3r-1} \bigg[ B_i^{r}D_L^{-4r}|\thhi - \bxi^T\bb|^{2r}
+ 3D_L^{-3r} [(1 + D_U D_L^{-1})(\ma h_i ) \util^TV^{-1}\util]^r \bigg].
\end{align}

From equation (\ref{first_der}) and since $\util^TV^{-1}\util \leq \bu V^{-1} \bu,$
it follows that
\begin{align*}
\sup_{\su \geq 0} \left|\dfrac{\partial^2 \thbi}{\partial (\su)^2} \right|
&\leq 
\sup_{\su \geq 0} \lb
4^{3r-1} \lp
B_i^rD_L^{-4r}|\thhi - \bxi^T\bb|^{2r}
+ 3 \lb D_L^{-4r}(1 + D_UD_L^{-1})(\ma h_i) \bu^TV^{-1}\bu\rb^r
\rp
\rb\\
&\leq 4^{3r-1}D_L^{-4r}\lb
\sup_{\su \geq 0} B_i^r|\thhi - \bxi^T\bb|^{2r}
+ 3 D_L^r (1 + D_UD_L^{-1})^r(\ma h_i)^r
\sup_{\su \geq 0} (\bu^TV^{-1}\bu)^r
\rb\\
&\leq 4^{3r-1}D_L^{-4r}\lb
D_U^r \sup_{\su \geq 0}\left|\frac{\thhi - \bxi^T\bb}{(\su + D_i)^{1/2}}\right|^{2r}
+ 3 (D_L + D_U)^r (\ma h_i)^r
\sup_{\su \geq 0}(\bu^TV^{-1}\bu)^r \rb
\end{align*}
where for all $\su \geq 0,$ 
$\left|\dfrac{\thhi - \bxi^T\bb}{(\su + D_i)^{1/2}}\right|^{2r} 
\stackrel{d}{=} |Z|^{2r}$ and $Z \sim N(0,1).$ Also, for all $\su \geq 0,$ $(\bu^TV^{-1}\bu)^r \stackrel{d}{=}
W_m^r,$ where $W_m \sim \chi_m^2.$ This implies that
\begin{align*}
E\lb\sup_{\su \geq 0}\left|\dfrac{\partial^2 \thbi}{\partial (\su)^2}\right|^{2r} \rb
&\leq
4^{3r-1}D_L^{-4r}\lb
D_U^r E[|Z|^{2r}] + 3 (D_L + D_U)^r (\ma h_i)^r E[W_m^{r}]
\rb = O(1).
\end{align*}
\end{proof}
\end{notreport}

Recall that $\bu = \that - X\bb \sim N(0,V).$ The results below then follow.

\begin{lem} 
\label{lem_comb2}
Let $r>0$ and assume $\max_{1 \leq i \leq m} \bxi^T\bb = O(1).$ Then 
\begin{align*}
||\that - X\btil||^{2r}
 = O_p(m^r)
\quad  \text{and} \quad
E\left[
||\that - X\btil||^{2r}
\right] = O(m^r). 
\end{align*}
\end{lem}

%%lemi
%%lem21
\begin{notreport}
\begin{proof}[Proof of Lemma \ref{lem_comb2}]
Recall $\util = \that - X\btil.$
Then $\sig^{-1/2}\util \sim N(0, I).$ Let  $W = \util^T\sig^{-1}\util \sim \chi_m^2$ and observe 
$\util^T\sig^{-1}\util \geq \util V^{-1} \util
\geq \util^T (\su + D_U)^{-1} I \util = (\su + D_U)^{-1} ||\util ||^2.$ This implies that $||\util||^{2r} \leq (\su + D_U)^{r} W^{r} = O_p(m^r). $ Also, $$E||\util||^{2r} \leq E\lb(\su + D_U)^{r} W^{r}\rb = O(m^{r}).\qedhere$$
%\ref{that_xbtils}.
\end{proof}
\end{notreport}

%%%%%%%%%%%%% TRAB LEMMA - COV %%%%%%%%%%%

%%lem28
\begin{lem}
\label{trab}
Let $\bz\sim N_p(\bzero,\Sigma).$ For matrices $A_{p\times p}$ and $B_{p\times p},$ where $B$ symmetric, we have
\begin{itemize}
\item[(i)]$\text{Cov}(\bz^TA\bz,\bz^TB\bz)
= 2\text{tr}(A\Sigma B\Sigma)$
\item[(ii)]$\text{Cov}(\bz^TA\bz,(\bz^TB\bz)^2)
=
4\text{tr}(A\Sigma B\Sigma)\text{tr}(B\Sigma)
+8\text{tr}(A\Sigma B\Sigma B\Sigma).$
\end{itemize}
\end{lem}
\begin{notreport}
\begin{proof}[Proof of (i)]
%See \citet*{searle_1971}.
See \citet*[pg. 51]{searle_1971}
\end{proof}
\begin{proof}[Proof of (ii)]
First, let $\sig = I_p.$ By the Spectral Decomposition Theorem, define $D := PBP^T$, where $P$ is orthogonal and $D$ is diagonal with eigenvalues $\lambda_i$. Define $C := PAP^T$. 
We know that $\bz^TB\bz = \bz^TP^TDP\bz$
and $\bz^TA\bz = \bz^TP^TCP\bz.$ Also, since 
$\bz\sim N_p(\bzero,I_p)$ and $\bz \stackrel{d}{=} P\bz$, 
$
\text{Cov}\left(\bz^TC\bz,(\bz^TD\bz)^2\right)
=
\text{Cov}\left(\bz^TA\bz,(\bz^TB\bz)^2\right).
$ Then by the above and algebra, we can show
$$
E\left[(\bz^TD\bz)^2\right] 
=2\text{tr}(B^2) + \text{tr}(B)^2$$ and 
$E\left[\bz^TC\bz(\bz^TD\bz)^2\right]
=
8\text{tr}(AB^2)+2\text{tr}(A)\text{tr}(B^2)
+4\text{tr}(AB)\text{tr}(B)+\text{tr}(A)\text{tr}(B)^2.$
Hence, \begin{align}
\label{covsq}
\text{Cov}\left(\bz^TA\bz,(\bz^TB\bz)^2\right)
=8\text{tr}(A B^2)+4\text{tr}(A B)\text{tr}(B). 
\end{align}

Now we assume general $\sig$ and let $\uw=\Sigma^{-1/2}\bz\sim N_p(\bzero,I_p).$ By (\ref{covsq}), we observe that
\begin{align*}
&\text{Cov}(\bz^TA\bz,(\bz^TB\bz)^2)\\
&\qquad= \text{Cov}(w^T\sig^{1/2}A\sig^{1/2}w, (w^T\sig^{1/2}B\sig^{1/2}w)^2) \\
&\qquad=
4\text{tr}(A\Sigma B\Sigma)\text{tr}(B\Sigma)
+8\text{tr}(A\Sigma B\Sigma B\Sigma).\qedhere
\end{align*}
\end{proof}
\end{notreport}
%
%%%%%%%%%%%%%% expectation sq of sut - su %%%%%%%
%lem 4
%lem 11
%lem26
\begin{lem}
\label{sut_sq}
$E[(\sut -\su)^2] = 2(m-p)^{-2}\sum_{i=1}^m (\su + D_i)^2 + O(m^{-2}).$
\end{lem}

\begin{notreport}
\begin{proof}
Observe $m-p = \tr\{
I-P_X
\}$ and  define $d = \sum_iD_i(1-h_i) = \tr\{
(I-P_X)D
\},$ where $D = \text{Diag}\{D_i\}.$ Also,
recall $\util = \that - X\btil.$
Then
\begin{align*}
E[(\sut -\su)^2] &=
(m-p)^{-2}E \lb \left\{
\util^T(I-P_X)\util - \su(m-p) -d
\right\}^2 \rb\\
&= 
(m-p)^{-2}E \lb \left\{
\util^T(I-P_X)\util - 
\tr\{
(I-P_X)V
\}
\right\}^2 \rb\\
&=(m-p)^{-2}\lb
E \bigg[ \left\{
\util^T(I-P_X)\util
\right\}^2 \rb \\
&-2\tr\{
(I-P_X)V
\}E[\util^T(I-P_X)\util]\\
&+ \tr\{
(I-P_X)V
\}^2
\bigg]
=2(m-p)^{-2}\tr\{
(I-P_X)V(I-P_X)V
\}.
\end{align*}
 Using matrix manipulations, it is easy to show that 
\begin{align*}
E[(\sut -\su)^2] &= 2(m-p)^{-2}\sum_{i=1}^m (\su + D_i)\bigg[
(\su + D_i)\\ &\qquad+ \bxi^T(X^TX)^{-1}X^TVX(X^TX)^{-1}\bxi^T 
-2(\su + D_i)^2h_{ii}^V
\bigg]\\
&= 2(m-p)^{-2}\sum_{i=1}^m (\su + D_i)^2 + O(m^{-2}).\qedhere
\end{align*}
\end{proof}
\end{notreport}
\begin{theorem}
If regularity conditions (i)--(iii) hold, then
\label{thm_mse}
$E[(\theblupi - \thi)^2] = g_{1i}(\su) + g_{2i}(\su) + 
g_{3i}(\su) + g_{4}(\su) + o(m^{-1}),$
where 
\begin{align*}
g_{1i}(\su) &= B_i \su\\
g_{2i}(\su) &= B_i^2 h_{ii}^V\\
g_{3i}(\su) &= B_i^3 D_i^{-1} \V(\sut)\\
g_{4}(\su) &= \sum_{i=1}^m  w_i^2 B_i^2 V_i - \sum_{i=1}^m \sum_{j=1}^m w_i w_j B_i B_j h_{ij}^V,
\end{align*}
and where $\V(\sut) = 2(m-p)^{-2}\sum_{k=1}^m(\su + D_k)^2 + o(m^{-1}).$

\end{theorem}

\begin{rem}
\label{rem1}
We note that the the MSE of the benchmarked EB estimator in Theorem 1 is always non-negative. It is clear that $g_{1i}(\su),$ $g_{2i}(\su),$ and $g_{3i}(\su)$ are non-negative. To establish the non-negativity of $g_4(\su),$ let $\boldsymbol{q} = (\boldsymbol{q}_1,\ldots,\boldsymbol{q}_m),$ where $q_i = w_iB_iV_i^{1/2}.$ We can write 
$g_4(\su) = \boldsymbol{q}^T(I-\tilde{P}^T_X)\boldsymbol{q},$ where
$\tilde{P}^T_X = V^{-1/2}X(X^TV^{-1}X)^{-1}X^{T}V^{-1/2}.$ Thus, $g_4(\su) \geq 0,$ and hence, the MSE in Theorem 1 is always non-negative.
\end{rem}

%% file: estimation_NEW.tex
\setcounter{chapter}{4}
\setcounter{equation}{0}
\noindent {\bf 4. Estimator of MSE Approximation}
%\mysection{Estimator of MSE Approximation}
\label{s_est}

We now obtain an estimator of the MSE approximation for the Fay-Herriot model (assuming normality). Theorem \ref{thm_est} shows that the expectation of the MSE estimator is correct up to $O(m^{-1}).$
\begin{lem}
\label{lit_lem}
Suppose that 
\begin{align}
\label{asm}
\sup_{t \in T} |h^\prime(t)| = O(m^{-1})
\end{align}
for some interval $T \subseteq \mathbb{R}$. If $\suh, \su \in T$ w.p. $1,$ then
$E[h(\suh)] = h(\su) + o(m^{-1}).$ 
\begin{proof}
Consider the expansion 
$h(\suh) = h(\su) + h^\prime (\sus)(\suh - \su)$ for some $\sus$  between $\su$ and $\suh.$  Then $\sus\in T\, \text{ a.s.},$
and
$h^\prime (\sus) \leq \sup_{t \in T} |h^\prime (t)|\,\text{ a.s.}$ as well.
This implies 
$E[h^\prime (\sus)(\suh - \su)] \leq 
\sup_{t \in T} |h^\prime (t)| E|\suh - \su| = O(m^{-3/2})$
by equation (\ref{asm}) and since $E|\suh - \su| \leq E^{\frac{1}{2}}[(\suh - \su)^2].$
Hence, if (\ref{asm}) holds, then $E[h(\suh)] = h(\su) + o(m^{-1}).$
\end{proof}
\end{lem}

\begin{theorem}
\label{thm_est}
%$E[(\theblupi - \thhi)^2] = $
$E[g_{1i}(\suh) + g_{2i}(\suh) + 2g_{3i}(\suh) + g_{4}(\suh)]
= g_{1i}(\su) + g_{2i}(\su) + g_{3i}(\su) +g_{4}(\su) + o(m^{-1}),$
where 
$g_{1i}(\su), g_{2i}(\su),  g_{3i}(\su),$ and $g_{4}(\su)$ are defined in Theorem 1.
%\begin{align*}
%g_{1i}(\su) &= B_i \su\\
%g_{2i}(\su) &= B_i^2 h_i^V\\
%g_{3i}(\su) &= B_i^3 D_i^{-1} \V(\sut)\\
%g_{4}(\su) &= \sum_{i=1}^m  w_i^2 B_i^2 V_i - \sum_{i=1}^m \sum_{j=1}^m w_i w_j B_i B_j h_{ij}^V,
%\end{align*}
%where $\V(\sut) = 2(m-p)^{-2}\sum_{k=1}^m(\su + D_k)^2 + o(m^{-1}).$
\end{theorem}

\begin{proof}
By Theorem A.3 in \citet*{prasad_1990}, $E[g_{1i}(\suh) + g_{2i}(\suh) + 2g_{3i}(\suh)] =g_{1i}(\su) + g_{2i}(\su) + g_{3i}(\su) + o(m^{-1}).$
In addition, we consider $E[g_{4}(\suh)],$ where
$g_{4}(\su) = \sum_{i=1}^m  w_i^2 B_i^2 V_i - \sum_{i=1}^m \sum_{j=1}^m w_i w_j B_i B_j h_{ij}^V=: g_{41}(\su) + g_{42}(\su) .$ 
We first show that the derivatives of $g_{41}(\su)$ and $g_{42}(\su)$ satisfy (\ref{asm}). Let $T = [0,\infty).$ Consider
\begin{align*} 
\sup_{\su \geq 0} \left|\frac{\partial g_{41}(\su)}{\partial \su} \right|
&= \sup_{\su \geq 0}\sum_{i=1}^m w_i^2 B_i^2= O(m^{-1}).
\end{align*}

It can be shown that
$\dfrac{\partial B_iB_j}{\partial \su} =
-B_iB_j^2D_j^{-1} - B_i^2B_jD_i^{-1}$
and $(X^TV^{-1}X)^{-1} \leq (X^TV^{-2}X)^{-1} D_L^{-1}. $ Observe that
\begin{align*} 
\left|\frac{\partial g_{42}(\su)}{\partial \su} \right|
&\leq \sum_{i=1}^m \sum_{j=1}^m w_i w_j
\bigg[|B_iD_L^{-1}h_{ij}^V| + |B_jD_L^{-1}h_{ij}^V|\\
&+ 
B_iB_j\bxi^T(X^TV^{-1}X)^{-1}X^TV^{-2}X(X^TV^{-1}X)^{-1}\bxi
\bigg]\\
&\leq 3m^2 (\ma w_i)^2D_L^{-1} B_i(\su + D_U)(\ma h_i)\\
&\leq 3m^2 (\ma w_i)^2 D_L^{-1}D_U (\su + D_L)^{-1}(\su + D_U)(\ma h_i)\\
&= 3m^2 (\ma w_i)^2 D_L^{-1}D_U (1 + D_UD_L^{-1})(\ma h_i)
 = O(m^{-1}).
\end{align*}
This implies that $\displaystyle \sup_{\su \geq 0} \left|\frac{\partial g_{42}(\su)}{\partial \su} \right| = O(m^{-1}).$ Since the derivatives of $g_{41}(\su)$ and $g_{42}(\su)$ satisfy (\ref{asm}), we know that $E[g_{4}(\suh)] = g_{4}(\su) + o(m^{-1}).$\qedhere

\end{proof}

%% file: bootstrap.tex
\setcounter{chapter}{5}
\setcounter{equation}{0}
\noindent {\bf 5. Parametric Bootstrap Estimator of the MSE of the Benchmarked Empirical  Bayes Estimator}
%\section{Parametric Bootstrap Estimate of the Benchmarked \\Empirical  Bayes Estimate}

In this section, we extend the methods of \citet*{butar_2003} to find a parametric bootstrap estimator of the MSE of the benchmarked EB estimator. Under the proposed model, the expectation of the proposed measure of uncertainty of the benchmarked EB estimator is correct up to order $O(m^{-1}).$ 

To introduce the parametric bootstrap method, consider the model
\begin{align}
\label{boot}
\hat{\theta_i}^*|u_i^* &\stackrel{ind}{\sim} N(\bxi^T\btil + u_i^*,D_i)\notag \\
u_i^*  &\stackrel{ind}{\sim} N(0,\suh).
\end{align}
Following \citet*{butar_2003}, we use the parametric bootstrap twice. We first use it to estimate $g_{1i}(\su), 
g_{2i}(\su),$ and $g_4(\su)$ by correcting the bias of 
$g_{1i}(\suh), g_{2i}(\suh),$ and $g_4(\suh).$ We then use it again to estimate $E[(\thebi - \thbi)^2] = g_{3i}(\su) + o(m^{-1}).$

\citet*{butar_2003} derived a parametric bootstrap estimator for the MSE of the EB estimator under the \citet*{fay_1979} model. Using Theorem A.1 of their paper, they show that the bootstrap estimator $V_i^{\text{BOOT}}$ is
\begin{align}
\label{vboot}
V_i^{\text{BOOT}} &= 2[g_{1i}(\suh) + g_{2i}(\suh)]
- E_*\lb
g_{1i}(\sustt) + g_{2i}(\sustt)
\rb + E_*[(\thebs - \thebi) ^2], 
\end{align}
where $E_*$ denotes the expectation computed with respect to the model given in~(\ref{boot}), and
%where $\thbi = \thhi(Y_i,\bb, \su)$ and
%$\thebi = \thhi(Y_i,\bhat(\suh), \suh).$ 
$\thebs = (1-B_i(\sustt))\thhi + B_i(\sustt)\bxi^T\bhat.$
Following their work, we propose a parametric bootstrap estimator of the MSE of the benchmarked EB estimator that is a simple extension of (\ref{vboot}). 

We propose to estimate $g_{1i}(\su) + g_{2i}(\su) + g_{4}(\su)$
by 
$$2[g_{1i}(\suh) + g_{2i}(\suh) + g_{4}(\suh)]
- E_*\lb
g_{1i}(\sustt) + g_{2i}(\sustt) + g_{4}(\sustt)
\rb$$ and then to estimate $E[(\thebi - \thbi)^2]$ by 
$E_*[(\thebs - \thebi)^2 ].$
Thus, our proposed estimator of $\text{MSE}[\hat{\theta}_i^{\text{EBM1}}]$ is 
\begin{align*}
%\label{vbootb}
V_i^{\text{B-BOOT}}
&= 2[g_{1i}(\suh) + g_{2i}(\suh) + g_4(\suh)]
- E_*\lb
g_{1i}(\sustt) + g_{2i}(\sustt) + g_{4}(\sustt)
\rb \\
&+ E_*[(\thebs - \thebi)^2 ]. \notag
\end{align*}
%
%We now show that the expectation of $V_i^{\text{B-BOOT}}$ is correct up to $O(m^{-1}).$
%
\begin{theorem}
$E[V_i^{\text{B-BOOT}}] = \text{MSE}[\hat{\theta}_i^{\text{EBM1}}] + o(m^{-1}).$
\end{theorem}
\begin{proof}
First, by Theorem A.1 in \cite{butar_2003}, we note that
\begin{align*}
E_*[g_{1i}(\sustt)] &= g_{1i}(\suh) -g_{3i}(\suh) + o_p(m^{-1}),\\
E_*[g_{2i}(\sustt)] &= g_{2i}(\suh) +
o_p(m^{-1}), \;\; \text{and}\\
E_*[(\thebs - \thebi)^2 ] &= g_{5i}(\suh) +
o_p(m^{-1}),
\end{align*}
where 
$g_{5i}(\suh) = [B_i(\suh)]^4D_i^{-2}
\lp\thhi - \bxi^T\btil(\suh)\rp^2.$
Also, $E_*[g_4(\sustt)] = g_4(\suh) + o_p(m^{-1}),$ which follows along the lines of the proof of Theorem~A.2(b) of \citet*{datta_2000}. Applying these results and our Theorem 2, we find 
$$V_i^\text{B-BOOT} = g_{1i}(\suh) + g_{2i}(\suh) + g_{3i}(\suh) + g_4(\suh)   + g_{5i}(\suh) + o_p(m^{-1}).$$
This implies that 
$$E[V_i^\text{B-BOOT}] = g_{1i}(\su) + g_{2i}(\su) + g_{3i}(\su) + g_4(\su) + o(m^{-1})$$ since $E[g_{5i}(\suh)] =  g_{3i}(\su) + o(m^{-1})$ by \cite{butar_2003}, and by applying the results of \citet{prasad_1990}.
%We first note that $E_*[g_4(\sustt)] = g_4(\suh) + o_p(m^{-1}),$ which follows along the lines of the proof of Theorem~A.2(b) of \citet*{datta_2000}. 
%Applying this result and Lemma~A.1 from \citet*{butar_2003}, we find that
%$$V_i^{\text{B-BOOT}}
%= g_{1i}(\suh) + g_{2i}(\suh) + g_{3i}(\suh) + g_4(\suh)   + g_{5i}(\suh) + o_p(m^{-1}),$$
%where
%$$g_{5i}(\suh) = [B_i(\suh)]^4D_i^{-2}
%\lp\thhi - \bxi^T\btil(\suh)\rp^2
%E[(\suh - \su)^2]$$
%under the Fay-Herriot model.
%Also,  
%following \mbox{\citet*{butar_2003}}, we can show
%$E[g_{5i}(\suh)] = g_{3i}(\su) + o(m^{-1}).$
%%
%%
%Then $E\lb
%V_i^{\text{B-BOOT}}
%\rb = g_{1i}(\su) + g_{2i}(\su) + g_{3i}(\su) +  g_{4}(\su) + o(m^{-1}) 
%= \text{MSE}[\hat{\theta}_i^{\text{EBM1}}] + o(m^{-1}).$
\end{proof}

%% file: mse_applic.tex
\setcounter{chapter}{6}
\setcounter{equation}{0}
\noindent {\bf 6. Two Applications}
%\section{An Application}

In this section, we consider a data set and report on a simulation study in order to compare the performance of the estimator of the MSE of the benchmarked EB estimator and the parametric bootstrap estimator of the MSE of the benchmarked EB estimator. Tables and figures that result from this can be found in Appendix A.

We consider data from the Small Area Income and Poverty Estimates (SAIPE) program at the U.S. Census Bureau, which produces model-based estimates of the number of poor school-aged children (5--17 years old) at the national, state, county, and district levels. The school district estimates are benchmarked 
to the state estimates by the Department of Education to allocate funds under the No Child Left Behind Act of 2001. Specifically, we consider year 1997.
In the SAIPE program, the model-based state estimates are benchmarked to the national school-aged poverty rate using the benchmarked estimator in (\ref{fayest}). The number of poor school-aged children has been collected from the Annual Social and Economic Supplement (ASEC) of the Current Population Survey (CPS) from 1995 to 2004, while the American Community Survey (ACS) estimates have been used since 2005. Additionally, the model-based county estimates are benchmarked to the model-based state estimates using the 
the benchmarked estimator in (\ref{fayest}).
%Consider the EB estimators as given in  (\ref{fayest}).
%Our modeling assumptions require that $\phi_i = w_i$ and $t = \sum_{i=1}^m w_i \thhi.$\\

In the SAIPE program, the state model for poverty rates in school-aged children follows the basic Fay-Herriot (1979) framework where
$\thhi = \theta_i + e_i$ and $\theta_i = \bxi^T\bb + u_i.$ Here
$\theta_i$ is the true state level poverty rate, ö 
$\thhi$ is the direct survey estimate (from 
CPS ASEC), $e_i$ is the sampling error term with assumed known variance $D_i >0$, $\bxi$ are 
the predictors, $\bb$ is the unknown vector of regression coefficients, and $u_i$ is the model error with unknown variance $\su.$ 
The explanatory variables in the model are the IRS income 
tax--based pseudo-estimate of the child poverty rate, IRS non-filer rate, food stamp 
rate, and the residual term from the regression of the 1990 Census estimated child 
poverty rate. We estimate $\bb$ using the weighted least squares type estimator $\btil(\suh) = 
\xg V^{-1}\that,$ and we estimate $\su$ using the modified moment estimator $\suh$ from Section 2. 

As shown in Table~A.1, the estimated MSE of the EB estimator, $\text{mse}(\thebi)$,  compared to the estimated MSE of the benchmarked EB estimator, $\text{mse}(\hat{\theta}_i^{EBM1}),$ differs by the constant $g_{4}(\su),$  $0.025.$ This constant is effectively the increase in MSE that we suffer from benchmarking, and we see that in this case it is small (compared to the values of the MSEs). Generally speaking, it is expected to be small since $g_{4}(\su) = O(m^{-1}).$ 

%It should be noted that in the proofs of our paper above and in \cite{prasad_1990}, we take advantage of the fact that 
%$$P(\sut \leq 0) = O(m^{-r}) \; \forall \; r > 0.$$
%Practically speaking, this fact should correspond to the aforementioned probability being very small. However, if this is not the case for some particular dataset with fixed $m,$ then our theoretical derivations of the MSE of the benchmarked EB estimator and the derivations in \cite{prasad_1990} under the Fay-Herriot model may be unreliable. We illustrate this concept with SAIPE dataset from 1997. 

In Table~A.1, we write $\text{mse}^B$ and $\text{mse}^{BB}$ as the bootstrap estimates of the MSE of the EB estimator and the benchmarked EB estimator, respectively.
As mentioned, we consider year 1997 for illustrative purposes.
%  
%We consider two years for illustrative purposes from the SAIPE dataset (years 1997 and 2000). Table 6.1 illustrates when $\suh$ is not too close to zero, being $3.08.$
When we performed the bootstrapping, we resampled $\sust$ $10,000$ times in order to calculate $\text{mse}^B$ and $\text{mse}^{BB}.$ 
%and the proportion of resampled values of $\sust$ that are negative is $0.034.$ Since this proportion is small, we can see that the bootstrap estimate of the MSE of the benchmarked EB estimator and the EB estimator
%are somewhat close to each other. 
%
This is best understood through the concept behind our bootstrapping approach. Consider the behavior of $g_{1i}(\su)$, the only term that is $O(1).$ Ordinarily, $g_{1i}(\suh)$ underestimates 
$g_{1i}(\su)$, and $E_*[g_{1i}(\suh)]$ underestimates 
$g_{1i}(\suh).$ The basic idea is that we use the amount by which $E_*[g_{1i}(\suh)]$ underestimates $g_{1i}(\suh)$ as an approximation of the amount by which $g_{1i}(\suh)$ underestimates $g_{1i}(\su).$

We run into a problem with the 1997 data, where $g_{1i}(\suh)$ is 0, since in this case $E_*[g_{1i}(\suh)]$ overestimates 
$g_{1i}(\suh).$ Recall that 
$$V_i^{\text{B-BOOT}} = g_{1i}(\suh) + \{
g_{1i}(\suh) - E_*[g_{1i}(\hat{\sigma}_u^{*2})]
\} + O(m^{-1}).$$
Since $g_{1i}(\suh)$ is 0 and is the dominating term of 
$V_i^{\text{B-BOOT}},$ many of the estimated MSEs of the benchmarked bootstrapped estimator ($\text{mse}^{\text{BB}}$) are negative. Also, observe this same behavior holds true for the bootstrapped estimator proposed by \cite{butar_2003}, which we denote by $\text{mse}^{\text{B}}.$ Hence, we do not recommend using bootstrapping when $\suh$ is too close to zero because of the form of $\suh$. We also note that the MSE of the benchmarked EB estimator is always non-negative as explained in Remark 1 of Section 3.
%However, in the very extreme case, we can look at the SAIPE data from year $1997,$ where our estimate $\suh$ is 0. In this case, we immediately notice from Table~2.2 that we have negative estimates of the MSE of the bootstrapped estimates. This occurs because the Prasad and Rao estimator is not an optimal choice when $\suh$ is too close to zero or even acutally $0$ in this case. This can easily been seen for two reasons. In the theoretical derivation, when $\suh$ is $0$, $g_{1i}(\suh)$ is 0. This also causes a component of 
%$g_{3i}(\suh)$ and $g_{4}(\suh)$ to be 0 (although not the entire component). Another point we stress is that when we resample the values of $\sust,$ they can be negative, which they are in this case because the value taken for $\sut$ can be 0 or even negative. This results in the unwanted behavior of the bootstrapped estimates having negative variances. Therefore, in our opinion, the Prasad and Rao approach should be used with caution when the moment estimator $\sut$ is too close to zero since estimates may be unreliable. 

%%%%%%%%%SIMULATION STUDY
In the second example, we ran a simulation study, using the
same covariates from the SAIPE dataset from 1997. We generated our data from the model
\begin{align}
\label{p0}
\hat{\theta}_i|\theta_i &\stackrel{ind}{\sim} N(\theta_i, D_i)\\
\theta_i &\stackrel{ind}{\sim} N(X^T\bb, \su),\notag
\end{align}
where $D_i$ comes from the SAIPE dataset. 
We first simulated 10,000 sets of values for $\theta_i$ and $\hat{\theta}_i$ using (\ref{p0}). We then used each set of $\hat{\theta}_i$ values as the data and computed the EB and benchmarked EB estimators according to  (2.3) and the EB formula given below it. 
In order to use EB, we took $\bb = (-3,0.5,1,1,0.5)^T$ and $\su = 5.$

In Figure~\ref{fig:sims}, we compare the estimator of the theoretical MSE of the benchmarked EB estimator and the 
bootstrap estimator of the MSE of the benchmarked EB estimator
with the true value, i.e., the average of the squared difference between the estimator values and the true $\theta_i,$ generated according to model (\ref{p0}). In the upper plot, 
we see that the estimator of the theoretical MSE of the benchmarked EB estimator overshoots the truth very slightly, which shows that our estimator is slightly conservative. We find the opposite behavior to be true of the bootstrap estimator of the MSE of the benchmarked Bayes estimator, meaning that it undershoots the truth slightly.

In practice, it seems safer to use a MSE estimator that overestimates than one that underestimates, and hence, we recommend our proposed MSE estimator over the bootstrapped MSE estimator.
%Since it is more conservative to overestimate the true parameter of interest, it is our opinion that the theoretical MSE is the better and more conservative choice.
%
Using the lower plot, we compared the theoretical Prasad Rao (PR) MSE estimator with the associated true value. 
We find the same behavior in the PR estimator as we did in our proposed theoretical MSE of the benchmarked EB estimator. The overshoot occurs in the terms that the estimators have in common, i.e., $g_{1i}(\su); g_{2i}(\su);$ and $
g_{3i}(\su).$ We see that for this particular simulation study where $m$ is particularly large at 10,000, the difference between the two MSEs is indistinguishable. 

%% file: mse_summary.tex
\setcounter{chapter}{7}
\setcounter{equation}{0}
\noindent {\bf 7. Summary and Conclusion}
%\section{Summary and Conclusion}

We have shown that the increase in MSE due to benchmarking under our modeling assumptions is quite small for the Fay-Herriot model, specifically $O(m^{-1}).$ We have derived an asymptotically unbiased estimate of the MSE of the benchmarked EB estimator (EBLUP) under the same assumptions which is correct to order $O(m^{-1}).$ We have derived a parametric bootstrap estimator of the benchmarked EB estimator based on work done by \citet*{butar_2003}.
Furthermore, we have illustrated our methodology for a data set for fixed $m$ using U.S. Census data. Since our theoretical estimator of the MSE under benchmarking is guaranteed to be positive, we recommend it over the one derived by bootstrapping. We also performed a simulation study that suggests use of the theoretical estimator of the MSE under benchmarking. 
%
%In closing, benchmarking model-based small-area estimates to match the aggregate of larger areas is becoming more prevalent to protect against models that are possibly misspecified. This is often also a political necessity to convince the legislators of the utility of small area estimates. 
In closing, it is important to pursue further work for more complex models, and, in particular, when it is necessary to achieve multi-stage benchmarking. \\
%We believe that the present work will be the genesis for future work in this direction. \\

%In closing, benchmarking model-based small-area estimates to match the aggregate of larger areas is becoming more prevalent in present times to protect against models that are possibly misspecified. This is often also a political necessity to convince the legislators of the utility of small area estimates. In this paper, we have found MSEs of benchmarked EB small area estimates in the special case of Fay-Herriot models. It is important to pursue this work for more complex models, and in particular, when it is necessary to achieve multistage benchmarking. We believe that the present work will be the genesis for future work in this direction. \\

\noindent {\large\bf Acknowledgment}  \par
The research was partially supported by NSF Grant SES 1026165 and the United States Census Bureau Dissertation Fellowship Program. The views expressed here are those of the authors and do not reflect those of the U.S. Census Bureau. We would like to express our thanks to an associate editor, the referees, and Professor J.N.K. Rao for their helpful suggestions. 

%% file: appendixa.tex
\newpage
\renewcommand{\theequation}{A.\arabic{equation}}
\renewcommand{\thefigure}{A.\arabic{figure}}
\renewcommand{\thetable}{A.\arabic{table}}
\section*{Appendix A}\begin{center}
%\singlespacing
\begin{longtable}{ccccccccc}
%\label{boot2}
\caption{Table of estimates for 1997}\\
\hline \hline
$i$	&	$\thhi$	&	$\thebi$	&$\theblupi$	&	$\text{mse}(\thhi)$	&	$\text{mse}(\thebi)$	&	$\text{mse}(\theblupi)$	&	$\text{mse}^B$	&	$\text{mse}^{BB}$\\\hline \hline
\endfirsthead
\caption{Table of estimates for 1997 (continued)} \\
\hline \hline
$i$	&	$\thhi$	&	$\thebi$	&$\theblupi$	&	$\text{mse}(\thhi)$	&	$\text{mse}(\thebi)$	&	$\text{mse}(\theblupi)$	&	$\text{mse}^B$	&	$\text{mse}^{BB}$\\
\hline \hline
\endhead
\hline\hline \multicolumn{8}{l}{}\endfoot 
%\width{1em}
1	&	25.16	&	21.38	&	21.56	&	15.72	&	1.38	&	1.41	&	0.02	&	0.04	\\
2	&	10.99	&	14.94	&	15.11	&	10.44	&	2.12	&	2.14	&	0.66	&	0.68	\\
3	&	23.35	&	20.89	&	21.06	&	11.84	&	1.68	&	1.70	&	0.00	&	0.01	\\
4	&	23.32	&	22.18	&	22.35	&	13.85	&	1.90	&	1.92	&	0.37	&	0.38	\\
5	&	23.55	&	22.71	&	22.88	&	2.39	&	5.92	&	5.94	&	1.12	&	1.13	\\
6	&	9.14	&	13.12	&	13.29	&	6.38	&	2.19	&	2.22	&	0.36	&	0.38	\\
7	&	10.34	&	13.39	&	13.56	&	9.85	&	2.08	&	2.10	&	0.39	&	0.41	\\
8	&	15.54	&	13.06	&	13.23	&	17.56	&	0.91	&	0.94	&	-0.47	&	-0.45	\\
9	&	35.85	&	32.43	&	32.60	&	32.35	&	4.92	&	4.95	&	3.49	&	3.50	\\
10	&	18.34	&	19.59	&	19.76	&	3.70	&	3.71	&	3.74	&	0.40	&	0.41	\\
11	&	23.52	&	20.53	&	20.70	&	12.93	&	1.16	&	1.19	&	-0.38	&	-0.37	\\
12	&	18.98	&	13.72	&	13.89	&	20.87	&	2.45	&	2.48	&	1.24	&	1.26	\\
13	&	17.56	&	13.64	&	13.82	&	12.38	&	1.70	&	1.73	&	0.23	&	0.25	\\
14	&	14.57	&	15.72	&	15.89	&	3.56	&	3.45	&	3.47	&	-0.06	&	-0.05	\\
15	&	11.07	&	12.53	&	12.70	&	7.58	&	1.84	&	1.86	&	-0.23	&	-0.22	\\
16	&	11.09	&	11.21	&	11.38	&	8.49	&	1.74	&	1.76	&	-0.24	&	-0.22	\\
17	&	11.01	&	13.48	&	13.65	&	9.34	&	1.61	&	1.63	&	-0.15	&	-0.14	\\
18	&	23.12	&	20.78	&	20.95	&	13.98	&	1.37	&	1.40	&	-0.12	&	-0.11	\\
19	&	21.08	&	24.15	&	24.32	&	15.19	&	1.80	&	1.82	&	0.40	&	0.42	\\
20	&	13.18	&	12.44	&	12.61	&	13.63	&	2.09	&	2.11	&	0.56	&	0.57	\\
21	&	9.90	&	13.16	&	13.33	&	9.28	&	1.65	&	1.67	&	-0.03	&	-0.01	\\
22	&	19.66	&	14.38	&	14.56	&	7.66	&	2.46	&	2.48	&	1.02	&	1.04	\\
23	&	13.78	&	16.86	&	17.03	&	4.04	&	3.11	&	3.13	&	0.38	&	0.39	\\
24	&	14.34	&	10.11	&	10.28	&	9.91	&	1.64	&	1.67	&	0.16	&	0.17	\\
25	&	20.58	&	22.30	&	22.47	&	15.07	&	2.42	&	2.45	&	0.97	&	0.99	\\
26	&	18.90	&	15.11	&	15.28	&	15.24	&	1.00	&	1.03	&	-0.37	&	-0.35	\\
27	&	17.00	&	18.60	&	18.77	&	12.95	&	1.37	&	1.40	&	-0.21	&	-0.19	\\
28	&	9.72	&	9.62	&	9.79	&	7.18	&	2.24	&	2.26	&	0.09	&	0.10	\\
29	&	14.06	&	12.94	&	13.12	&	10.23	&	1.71	&	1.74	&	-0.06	&	-0.04	\\
30	&	10.94	&	6.72	&	6.89	&	11.35	&	1.88	&	1.91	&	0.50	&	0.52	\\
31	&	14.66	&	13.28	&	13.45	&	5.52	&	2.48	&	2.51	&	-0.03	&	-0.01	\\
32	&	29.69	&	24.44	&	24.61	&	13.18	&	2.62	&	2.65	&	1.38	&	1.40	\\
33	&	23.76	&	22.85	&	23.02	&	3.10	&	4.76	&	4.79	&	0.94	&	0.95	\\
34	&	13.90	&	16.58	&	16.75	&	5.70	&	2.29	&	2.31	&	-0.01	&	0.01	\\
35	&	18.19	&	13.64	&	13.81	&	11.92	&	1.81	&	1.84	&	0.48	&	0.50	\\
36	&	13.91	&	13.64	&	13.81	&	3.95	&	3.07	&	3.10	&	-0.25	&	-0.23	\\
37	&	16.09	&	21.50	&	21.68	&	11.14	&	1.52	&	1.54	&	0.24	&	0.26	\\
38	&	12.60	&	13.43	&	13.60	&	10.35	&	2.53	&	2.56	&	0.83	&	0.84	\\
39	&	14.61	&	13.92	&	14.09	&	3.73	&	3.40	&	3.42	&	-0.01	&	0.00	\\
40	&	20.37	&	14.60	&	14.77	&	18.53	&	1.04	&	1.07	&	-0.15	&	-0.14	\\
41	&	18.74	&	21.21	&	21.38	&	14.57	&	1.49	&	1.52	&	0.02	&	0.04	\\
42	&	12.87	&	15.77	&	15.94	&	12.94	&	1.98	&	2.01	&	0.46	&	0.47	\\
43	&	16.09	&	16.10	&	16.27	&	11.94	&	1.92	&	1.95	&	0.28	&	0.30	\\
44	&	21.95	&	21.38	&	21.55	&	3.38	&	4.05	&	4.07	&	0.38	&	0.40	\\
45	&	11.27	&	9.76	&	9.93	&	9.45	&	2.28	&	2.31	&	0.50	&	0.51	\\
46	&	11.15	&	10.10	&	10.27	&	11.95	&	2.45	&	2.48	&	0.86	&	0.88	\\
47	&	16.40	&	14.96	&	15.13	&	11.51	&	1.20	&	1.22	&	-0.49	&	-0.47	\\
48	&	12.26	&	13.17	&	13.34	&	9.33	&	1.85	&	1.87	&	0.01	&	0.02	\\
49	&	18.76	&	22.25	&	22.42	&	13.73	&	3.81	&	3.83	&	2.46	&	2.48	\\
50	&	7.60	&	11.87	&	12.04	&	6.41	&	2.74	&	2.76	&	0.97	&	0.98	\\
51	&	11.74	&	11.70	&	11.87	&	8.86	&	2.08	&	2.10	&	0.17	&	0.19	\\
\end{longtable}
%\doublespacing
\end{center}

\begin{figure}[htbp]
\begin{center}
\includegraphics{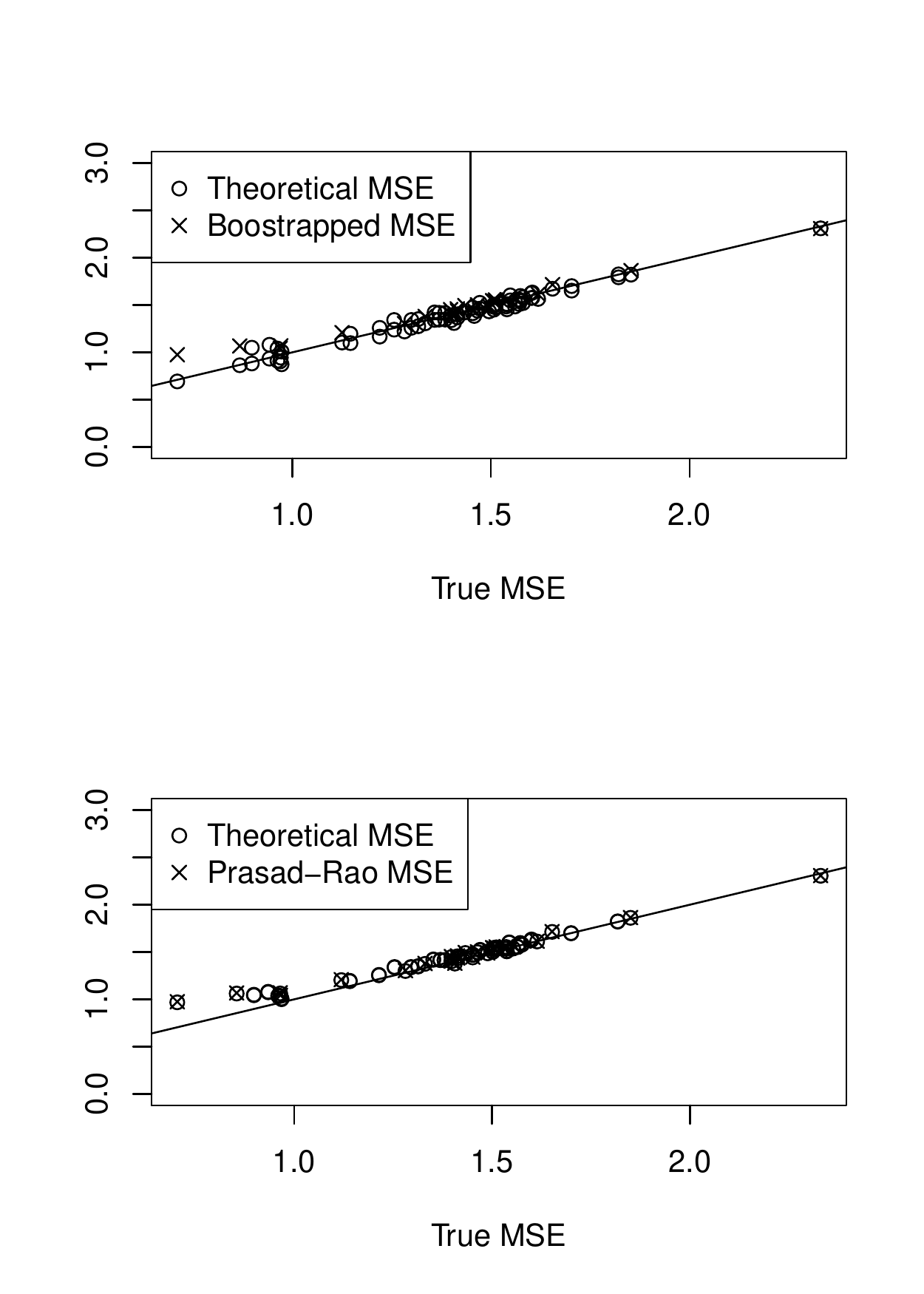}
\caption{Comparing Simulated MSEs with True MSEs}
\label{fig:sims}
\end{center}
\end{figure}

%% file: appendixb.tex
\newpage
\renewcommand{\theequation}{B.\arabic{equation}}
\section*{Appendix B} 
\begin{proof}[Proof of Theorem 1]
Observe that
\begin{align}
\label{bigeqn}
E[(\theblupi - \thi)^2] &= 
E[(\thbi -\thi)^2] + E[(\theblupi - \thbi)^2] \notag\\
&= E[(\thbi - \theta_i)^2] + E[(\thbi - \thebi - t + \thbarweb )^2] \notag\\
&= E[(\thbi - \theta_i)^2] + E[(\thbi - \thebi + \thbarweb - \thbarwb + \thbarwb - t)^2] \notag\\
&= E[(\thbi - \theta_i)^2] + E[ (\thebi-\thbi)^2] + E[(\thbarwb - t)^2] + E[(\thbarweb - \thbarwb)^2]  \notag \\ &-2E[(\thebi - \thbi )(\thbarweb - \thbarwb)] -
 \;  2E[(\thebi - \thbi)(\thbarwb - t)] \notag \\
&+ 2E[(\thbarweb - \thbarwb)(\thbarwb - t)].
\end{align}

Next, observe that $E[(\thbi - \theta_i)^2] + E[ \thebi-\thbi]^2 = g_{1i}(\su) + g_{2i}(\su) + g_{3i}(\su) + o(m^{-1}),$ by \mbox{\citet*{prasad_1990}}, where 
\begin{align*}
g_{1i}(\su) &= B_i \su\\
g_{2i}(\su) &= B_i^2 h_{ii}^V\\
g_{3i}(\su) &= B_i^3 D_i^{-1} \V(\sut).
\end{align*}
%where $\V(\sut) = 2(m-p)^{-2}\sum_{k=1}^m(\su + D_k)^2 + o(m^{-1}).$\\
It may be noted that while $g_{1i}(\su) = O(1),$ both
$g_{2i}(\su)$ and $g_{3i}(\su)$ are of order $O(m^{-1}),$ as shown in \citet*{prasad_1990}. We show that 
$E[(\thbarwb - t)^2]= g_4(\su) = O(m^{-1}),$ whereas the remaining four terms of expression (\ref{bigeqn}) are of order $o(m^{-1}).$

%%%%%%%%%%%%%%%% TERM 3 %%%%%%%%%%%%%%%%%%%%%%%%%%%%%
First, we show that $E[(\thbarwb - t)^2] = g_4(\su).$ We
write $\thbarwb - t = -\sum_{i=1}^mw_iB_i(\thhi -\bxi^T\btil)$ and consider
\begin{align}
\label{g4}
E[(\thbarwb - t)^2]
&=E\lb \left\{\sum_{i=1}^m  w_i B_i(\thhi -\bxi^T\btil)\right\}^2\rb\notag \\
&= \sum_{i=1}^m  w_i^2 B_i^2 E[(\thhi -\bxi^T\btil)^2]
+  \sum_{i \neq j}  w_i w_j B_i B_j 
E[(\thhi -\bxi^T\btil)(\thhj -\bxj^T\btil)]\notag\\
&= \sum_{i=1}^m  w_i^2 B_i^2 (V_i - h_{ii}^V)
+  \sum_{i \neq j}  w_i w_j B_i B_j (-h_{ij}^V)\notag\\
&= \sum_{i=1}^m  w_i^2 B_i^2 V_i - \sum_{i=1}^m \sum_{j=1}^m w_i w_j B_i B_j h_{ij}^V.
\end{align}
Note that the expression on the right hand side of 
(\ref{g4}) is $O(m^{-1})$ since $\ma h_{ii} = O(m^{-1}),$ which implies that 
$\max_{1 \leq i \leq j \leq m} h_{ij}^V = O(m^{-1}).$

%
%where we note that $\max_{1 \leq i \leq j \leq m} h_{ij}^V = O(m^{-1})$ because
%$\ma h_{ii} = O(m^{-1}).$\\
%

%%%%%%%%%%%%%%%%%%%%%%%%%%%%%%%
%%%%%%%%%% TERM 2 %%%%%%%%%%%%%%
Next, we return to (\ref{bigeqn}) and show that \mbox{$E[(\thbarweb - \thbarwb)^2] = o(m^{-1}).$} 
Consider that
\begin{align}
\label{thbeb_thbw}
&E[(\thbarweb - \thbarw)^2] 
 = \sum_i w_i^2E\left[  (\thebi-\thbi)^2\right] + 2\sum_{i=1}^{m-1} \sum_{j=i+1}^m w_i w_jE\left[ 
(\thebi-\thbi)
(\thebj-\thbj)\right]\notag\\
&= 2\sum_{i=1}^{m-1} \sum_{j=i+1}^m w_i w_jE\left[ 
(\thebi-\thbi)
(\thebj-\thbj)\right] + o(m^{-1}),
\end{align}
since $\sum_i w_i^2 E[(\thebi - \thbi)^2] = o(m^{-1}). $ The latter holds because $ E[(\thebi - \thbi)^2] = g_{2i}(\su) + g_{3i}(\su) = O(m^{-1}),$
$\ma w_i = O(m^{-1}),$ and $\sum_i w_i = 1.$
Thus, it suffices to show $E\left[ 
(\thebi-\thbi)
(\thebj-\thbj)\right] = o(m^{-1})$ for all $i \neq j,$ and we do so by expanding $\thebi$ about $\thbi.$
For simplicity of notation, denote $\dfrac{\partial \thbi}{\partial \su} = 
\dfrac{\partial \thbi(\su)}{\partial \su}$ and 
$\dfrac{\partial^2 \thbis}{\partial (\su)^2} = \dfrac{\partial^2 \thbi(\sus)}{\partial (\su)^2}.$
Then
$$ \thebi - \thbi = \frac{\partial \thbi}{\partial \su}(\suh - \su) + \frac{1}{2} \frac{\partial^2 \thbis}{\partial (\su)^2}(\suh - \su)^2$$ for some $\sus$ between $\su$ and $\suh.$
 The expansion of 
$\thebj$ about $\thbj$ is similar.
%%%

Consider $E[(\thebi-\thbi)(\thebj-\thbj)]$ for 
$ i \neq j.$
Notice that
\begin{align}
E[(\thebi-\thbi)(\thebj-\thbj)]&= E\left[
\frac{\partial \thbi}{\partial \su}\frac{\partial \thbj}{\partial \su} (\suh - \su)^2
\right] 
+ \frac{1}{2}E\left[\frac{\partial \thbi}{\partial \su}
\frac{\partial^2 \thbjs }{\partial (\sigma_u^{2})^2} (\suh - \su)^3\right] \notag \\ 
&+ \frac{1}{2}E\left[
\frac{\partial^2 \thbis}{\partial (\sigma_u^{*2})^2}
 \frac{\partial \thbj}{\partial \su}(\suh - \su)^3
\right] 
+ \frac{1}{4}E\left[\frac{\partial^2 \thbis}{(\partial \su)^2}
\frac{\partial^2 \thbjs }{\partial^2 (\sigma_u^{2})^2} (\suh - \su)^4\right] \notag \\
&:= R_0 + R_1 + R_2 + R_3.\notag
\end{align}
In $R_1,$ 
\begin{align}
\label{r1}
 E\left[\frac{\partial \thbi}{\partial \su}
\frac{\partial^2 \thbjs}{\partial (\su)^2}(\suh - \su)^3 \right] \notag
&= 
E\left[\frac{\partial \thbi}{\partial \su}
\frac{\partial^2 \thbjs}{\partial (\su)^2}(\sut - \su)^3 I(\sut > 0) \right] \\
&- E\left[\frac{\partial \thbi}{\partial \su}
\frac{\partial^2 \thbjs}{\partial (\su)^2} (\su)^3 I(\sut \leq 0) \right].
\end{align}
Observe that
\begin{align*}
E\left[\frac{\partial \thbi}{\partial \su}
\frac{\partial^2 \thbjs}{\partial (\su)^2} (\su)^3 I(\sut \leq 0) \right] 
&\leq \sigma_u^6
E^\frac{1}{4}\lb\left\{\frac{\partial \thbi}{\partial \su}\right\}^4\rb
E^\frac{1}{4}\lb\left\{\frac{\partial^2 \thbjs}{\partial (\su)^2}\right\}^4\rb
P^\frac{1}{2}(\sut \leq 0)\\
&\leq
\sigma_u^6
E^\frac{1}{4}\lb 
\left
\{\frac{\partial \thbi}{\partial \su}\right\}^4\rb
E^\frac{1}{4}\lb 
\sup_{\su \geq 0}
\left\{ 
\frac{\partial^2 \thbj}{\partial (\su)^2}\right\}^4\rb
P^\frac{1}{2}(\sut \leq 0)\\
&= o(m^{-r})
\end{align*}
for all $r > 0$ by Lemmas \ref{part_thibs_sq} (ii) and \ref{lem_comb2}, which we have proved in Appendix A. Also,  $P(\sut \leq 0) = O(m^{-r}) \;\forall \; r>0,$ as proved in Lemma A.6 of \citet*{prasad_1990}.
Now
\begin{align}
E\left[\frac{\partial \thbi}{\partial \su}
\frac{\partial^2 \thbjs}{\partial (\su)^2}(\sut - \su)^3 I(\sut > 0) \right] \notag
&= E\left[\frac{\partial \thbi}{\partial \su}
\frac{\partial^2 \thbjs}{\partial (\su)^2}(\sut - \su)^3  \right]\\
&- E\left[\frac{\partial \thbi}{\partial \su}
\frac{\partial^2 \thbjs}{\partial (\su)^2}(\sut - \su)^3 I(\sut \leq 0) \right], \label{r1_1}
\end{align}
where  the second term expression in (\ref{r1_1}) is $O(m^{-r})$ since $P(\sut \leq 0) = O(m^{-r}) \;\forall \; r>0.$ We next observe that \begin{align*}
E\left[\frac{\partial \thbi}{\partial \su}
\frac{\partial^2 \thbjs}{\partial (\su)^2}(\sut - \su)^3  \right]
&\leq E^\frac{1}{4}\lb\left\{\frac{\partial \thbi}{\partial \su}\right\}^4\rb
E^\frac{1}{4}\lb\left\{
\frac{\partial^2 \thbjs}{\partial (\su)^2}
\right\}^4\rb
E^\frac{1}{2}[(\sut - \su)^6]\\
&\leq
E^\frac{1}{4}\lb\left\{\frac{\partial \thbi}{\partial \su}\right\}^4\rb
E^\frac{1}{4}\lb
\sup_{\su \geq 0}
\left\{
\frac{\partial^2 \thbj}{\partial (\su)^2}
\right\}^4\rb
E^\frac{1}{2}[(\sut - \su)^6] \\
&= O(m^{-3/2})
\end{align*}
since $E[(\sut - \su)^{2r}] = O(m^{-r})$ for any $r\geq 1$
by Lemma A.5 in \citet*{prasad_1990}.
This proves that $R_1 = o(m^{-1})$ since $\ma w_i = O(m^{-1}).$ By symmetry, $R_2$ is also $o(m^{-1}).$ Finally, we show that $R_3$ is 
$o(m^{-1}).$
Using a similar calculation involving $R_1,$ we can show that 
\begin{align}
\label{r3_1}
E\left[\frac{\partial^2 \thbis}{(\partial \su)^2}
\frac{\partial^2 \thbjs}{\partial^2 (\sigma_u^{2})^2}(\suh - \su)^4 \right] &=
E\left[\frac{\partial^2 \thbis}{(\partial \su)^2}
\frac{\partial^2 \thbjs}{\partial^2 (\sigma_u^{2})^2} (\sut - \su)^4\right] + o(m^{-r}).
\end{align}
Observe now that 
\begin{align*}
E\left[\frac{\partial^2 \thbis}{(\partial \su)^2}
\frac{\partial^2 \thbjs}{\partial^2 (\sigma_u^{2})^2} (\sut - \su)^4\right] &\leq
E^\frac{1}{4}\lb
\left\{
\frac{\partial^2 \thbis}{(\partial \su)^2}
\right\}^4\rb
E^\frac{1}{4}\lb\left\{
\frac{\partial^2 \thbjs}{\partial^2 (\sigma_u^{2})^2}
\right\}^4\rb
E^\frac{1}{2}\lb
(\sut - \su)^8
\rb\\
&\leq
E^\frac{1}{4}\lb
\sup_{\su \geq 0}
\left\{
\frac{\partial^2 \thbi}{(\partial \su)^2}
\right\}^4\rb
E^\frac{1}{4}\lb
\sup_{\su \geq 0}
\left\{
\frac{\partial^2 \thbj}{\partial^2 (\sigma_u^{2})^2}
\right\}^4\rb
E^\frac{1}{2}\lb
(\sut - \su)^8
\rb \\
 &= O(m^{-2}).
\end{align*}
Plugging this back into (\ref{r3_1}), we find that 
$E\left[\dfrac{\partial^2 \thbis}{(\partial \su)^2}
\dfrac{\partial^2 \thbjs}{\partial^2 (\sigma_u^{2})^2}(\suh - \su)^4 \right] = o(m^{-1}).$ Hence, $R_3$ is $o(m^{-1}).$ Finally, by calculations similar to those used for (\ref{r1}), we find that 
\begin{align*}
R_0 &= E\left[
\frac{\partial \thbi}{\partial \su}\frac{\partial \thbj}{\partial \su}(\suh - \su)^2 
\right] =
E\left[
\frac{\partial \thbi}{\partial \su}\frac{\partial \thbj}{\partial \su}(\sut - \su)^2 
\right] + o(m^{-r}).
\end{align*}
Take $\sig = V - X(X^TV^{-1}X)^{-1}X^T = (I-P_X^V)V,$ where $P_X = X(X^TV^{-1}X)^{-1}X^T,$ write $P_X^V = X(X^TV^{-1}X)^{-1}X^TV^{-1},$
and let $\bei$ be the $i$th unit vector.
We can show $\dfrac{\partial \thbi}{\partial \su} = B_i\bei^T\sig V^{-2}\util,$ where 
$\util = \that - X\btil.$ Define $A_{ij} = B_iB_jV^{-2}\sig\bei \bej^T\sig V^{-2} $and consider 
\begin{align*}
E\left[
\dfrac{\partial \thbi}{\partial \su}\dfrac{\partial \thbj}{\partial \su}(\sut - \su)^2 
\right]
&= E[\util^T A_{ij}\util (\sut - \su)^2]\\
&= \cov(\util^T A_{ij}\util, (\sut - \su)^2) + E[\util^T A_{ij}\util ]E[(\sut - \su)^2].
\end{align*}
Using Lemma \ref{trab} and the relation $(I-P_X)\sig = (I-P_X)V$,
\begin{align}
\label{cov_t4}
&\cov(\util^T A_{ij}\util, (\sut - \su)^2)\notag\\
 &\qquad=
(m-p)^{-2}\cov(\util^T A_{ij}\util, [\util^T(I-P_X)\util - \tr\{(I-P_X)V\}]^2)\\
&\qquad= (m-p)^{-2}\cov(\util^T A_{ij}\util, [\util^T(I-P_X)\util]^2 ) \notag\\
&\qquad-2(m-p)^{-2}\cov(\util^T A_{ij}\util, \util^T(I-P_X)\util ) \tr\{(I-P_X)V\}\notag\\
&\qquad= (m-p)^{-2}\bigg\{4\tr\{A_{ij}V(I-P_X)V \}\tr\{(I-P_X)V\}\notag\\ 
&\qquad+ 8\tr\{A_{ij}V(I-P_X)V (I-P_X)V\}\notag\\
&\qquad- 
4\tr\{A_{ij}V(I-P_X)V \}\tr\{(I-P_X)V\} \bigg\}\notag\\
&\qquad= 8(m-p)^{-2}\tr\{A_{ij}V(I-P_X)V (I-P_X)V\}.\notag\\
&\qquad= 8(m-p)^{-2}B_iB_j\bej^T \sig V^{-1}(I-P_X)V(I-P_X)V^{-1} \sig \bei, \notag
\end{align}
where $\text{tr}$ denotes the trace. 
Observe that $(I-P_X)V^{-1} \sig = I - (\pxv)^T$ and
$ (I - \pxv)V(I - (\pxv)^T) = \sig.$ Then
\begin{align*}
\cov(\util^T A_{ij}\util, (\sut - \su)^2) &= 
8(m-p)^{-2}B_iB_j\bej^T \sig V^{-1}(I-P_X)V(I-P_X)V^{-1} \sig \bei \\
&= 8(m-p)^{-2}B_iB_j\bej^T (I-\pxv)V(I - (\pxv)^T) \bei\\
&= 8(m-p)^{-2}B_iB_j\bej^T \sig \bei\\
&= 8(m-p)^{-2}B_iB_j\bej^T V \bei + O(m^{-3}) = O(m^{-3}),
\end{align*}
since the first term is zero because $i \neq j$ and $V$ is diagonal.
We now calculate $$E[\util^T A_{ij} \util] = 
\tr\{B_iB_j V^{-2} \sig \bei \bej^T \sig V^{-2} \sig \}
= B_iB_j \bej^T \sig V^{-2} \sig  V^{-2} \sig \bei.$$
Observe that $\sig V^{-2} \sig = 
I - (\pxv)^T - \pxv + \pxv(\pxv)^T.$ Then, after some computations, we find that $E[\util^T A_{ij} \util] =
B_iB_j \bej^T V^{-1} \bei + O(m^{-1}) = O(m^{-1})$ since $i \neq j.$
By Lemma \ref{sut_sq}, $E[(\sut - \su)^2] = 2(m-p)^{-2}\sum_{k=1}^m(\su + D_k)^2 + O(m^{-2}).$ 
Then
$$
E[\util^T A_{ij} \util]E[(\sut - \su)^2] = o(m^{-1}),$$
since $i \neq j.$
This implies that 
$R_0 = o(m^{-1}),$ which in turn implies that
\begin{align}
\label{ebi_bi_ebj_bj}
E[(\thebi - \thbi)(\thebj - \thbj)] = o(m^{-1}) \;\;\text{for}\;\; i \neq j,
\end{align}
since
$R_0, R_1, R_2,$ and  $R_3$ are all $o(m^{-1}).$
 Finally, this and (\ref{thbeb_thbw}) establishes that 
$
E[(\thbarweb - \thbarw)^2] = o(m^{-1}).$
%%%%%%%%%%%%% TERM 4 %%%%%%%%%%%%%%%%%%%%%%%%%%%%%%%%%%%

We return to (\ref{bigeqn}) to show that
$E[(\thebi - \thbi )(\thbarweb - \thbarwb)] = o(m^{-1}).$
By the Cauchy-Schwarz inequality, we find that
\begin{align*}
E[(\thebi - \thbi )(\thbarweb - \thbarwb)] 
&\leq
E^{\frac{1}{2}}\lb
(\thebi - \thbi )^2
\rb
E^{\frac{1}{2}}\lb
(\thbarweb - \thbarwb)^2
\rb = o(m^{-1}),
\end{align*}
since the first term is $O(m^{-1/2})$ and the second term is $o(m^{-1/2}).$

%%%%%%%%%%%%%%%%%%%% TERM 5 %%%%%%%%%%%%%%%%%%%%%%%%%%%%%%%%%%%
For the next term of (\ref{bigeqn}), we are interested in showing that $E[(\thebi - \thbi)(\thbarwb - t)] = o(m^{-1}).$ First, by Taylor expansion, we find that
$$\thebi - \thbi = 
\dfrac{\partial \thbi}{\partial \su} (\suh - \su)
+ \frac{1}{2}\dfrac{\partial^2 \thbis}{\partial (\su)^2}
(\suh - \su)^2$$ for some $\sus$ between $\su$ and $\suh.$
%+ \frac{1}{6}\dfrac{\partial^3 \thbis}{\partial (\su)^3}
%(\suh - \su)^3.$$ 
Consider that
$\thbarw - t = -\sum_iw_iB_i(\thhi - \bxi^T\btil).$ Then
\begin{align*}
E[(\thebi - \thbi)(\thbarwb - t)] & =
-\sum_jw_jB_jE\lb\dfrac{\partial \thbi}{\partial \su} (\suh - \su)(\thhj - \bxj^T\btil)\rb \\
&-\frac{1}{2}\sum_jw_jB_jE\lb\dfrac{\partial^2 \thbis}{\partial (\su)^2} (\suh - \su)^2(\thhj - \bxj^T\btil)\rb 
:= R_4 + R_5.
\end{align*}
Observe that
\begin{align}
\label{t5_ind}
E\lb\dfrac{\partial \thbi}{\partial \su} (\suh - \su)(\thhj - \bxj^T\btil)\rb
&= -\su E\lb\dfrac{\partial \thbi}{\partial \su}  (\thhj - \bxj^T\btil) I(\sut \leq 0)\rb \\
&+ E\lb\dfrac{\partial \thbi}{\partial \su}  (\sut - \su)(\thhj - \bxj^T\btil) I(\sut > 0)\rb \notag \\
&= E\lb\dfrac{\partial \thbi}{\partial \su}  (\sut - \su)(\thhj - \bxj^T\btil) I(\sut > 0)\rb + o(m^{-r}) \notag \\
&= E\lb\dfrac{\partial \thbi}{\partial \su}  (\sut - \su)(\thhj - \bxj^T\btil) \rb \notag \\ 
&- E\lb\dfrac{\partial \thbi}{\partial \su}  (\sut - \su)(\thhj - \bxj^T\btil) I(\sut \leq 0)\rb 
+ o(m^{-r}) \notag \\
&= E\lb\dfrac{\partial \thbi}{\partial \su}  (\sut - \su)(\thhj - \bxj^T\btil) \rb + o(m^{-r}) \notag
\end{align}
 since we may observe that $E\lb\dfrac{\partial \thbi}{\partial \su}  (\su) (\thhj - \bxj^T\btil) I(\sut \leq 0)\rb = o(m^{-r})$
and $E\lb\dfrac{\partial \thbi}{\partial \su}  (\sut - \su)(\thhj - \bxj^T\btil) I(\sut \leq 0)\rb 
= o(m^{-r}).$
Now, note that $
\dfrac{\partial \thbi}{\partial \su} = B_i\bei^T \sig V^{-2} \util, 
$ and write $D_{ij} = B_i V^{-2}\sig\bei \bej^T.$
Then by calculations similar to those in (\ref{cov_t4}), we find
\begin{align*}
E\lb\dfrac{\partial \thbi}{\partial \su}  (\sut - \su)(\thhj - \bxj^T\btil) \rb
&= \cov(\util^T D_{ij}\util, \sut - \su) \\
&= (m-p)^{-1}\cov(\util^T D_{ij}\util, \util^T(I-P_X)\util - \tr\{(I-P_X)V\})\\
&= 2(m-p)^{-1}\tr\{ 
D_{ij}V(I-P_X)V
\}\\
&= 2(m-p)^{-1}\tr\{ 
B_i V^{-2}\sig\bei \bej^T V(I-P_X)V
\}\\
&= 2(m-p)^{-1}
B_i\bej^T V(I-P_X)V^{-1}\sig\bei \\
&=2(m-p)^{-1} B_i \bej^T V(I - (P_X^V)^T)\bei \\
&= 2(m-p)^{-1} B_i [\bej^T V\bei - h_{ij}^V]\\
&= 2(m-p)^{-1} B_i \bej^T V\bei + o(m^{-1}).
\end{align*}
With this, we find that 
$$\sum_j w_j B_j E\lb\dfrac{\partial \thbi}{\partial \su}  (\sut - \su)(\thhj - \bxj^T\btil) \rb =2(m-p)^{-1}B_i^2 w_i (\su + D_i) + o(m^{-1})
= o(m^{-1}).$$ Hence, 
$R_4$ is $o(m^{-1}).$ 
We now show that $R_5 = o(m^{-1}).$ By calculations similar to those in (\ref{t5_ind}),
\begin{align*}
&\sum_j w_j B_j E\lb\dfrac{\partial^2 \thbis}{\partial (\su)^2} (\suh - \su)^2(\thhj - \bxj^T\btil)\rb \\
&\qquad=\sum_j w_j B_j E\lb\dfrac{\partial^2 \thbis}{\partial (\su)^2} (\sut - \su)^2(\thhj - \bxj^T\btil)\rb + o(m^{-r}).
\end{align*}
%%%
Recall that $E\lb \left\{
\sum_j w_j B_j (\thhj - \bxj^T\btil)
\right\}^2
\rb = O(m^{-1})$ by (\ref{g4}).
Now note that
\begin{align*}
&\sum_j w_j B_j E\lb\dfrac{\partial^2 \thbis}{\partial (\su)^2} (\sut - \su)^2(\thhj - \bxj^T\btil)\rb \\
& \qquad \leq  E^{\frac{1}{4}}
\lb \left\{
\dfrac{\partial^2 \thbis}{\partial (\su)^2}  
\right\}^4
\rb
E^{\frac{1}{4}}
\lb 
(\sut - \su)^8 
\rb
E^{\frac{1}{2}}
\lb \left\{
\sum_j w_j B_j (\thhj - \bxj^T\btil)
\right\}^2
\rb\\
&
\qquad \leq  E^{\frac{1}{4}}
\lb \left\{
\sup_{\sigma_u^2 \geq 0} \dfrac{\partial^2 \thbi}{\partial (\su)^2}  
\right\}^4
\rb
E^{\frac{1}{4}}
\lb 
(\sut - \su)^8 
\rb
E^{\frac{1}{2}}
\lb \left\{
\sum_j w_j B_j (\thhj - \bxj^T\btil)
\right\}^2
\rb \\
 & \qquad=O(m^{-3/2})
\end{align*}
by Lemma \ref{lem_comb1}(ii), by Theorem A.5 of \citet*{prasad_1990}, and by expression~(\ref{g4}). Thus, $R_5$ is $o(m^{-1}),$ and
 $E[(\thebi - \thbi)(\thbarwb - t)] = o(m^{-1}).$ 
 
%
% 
%%%%%%%%%%%%%%%%% TERM 6 %%%%%%%%%%%%%%%%%%%%%%
For the last term in (\ref{bigeqn}), we use the 
the Cauchy-Schwartz inequality to show
\begin{align*}
E[(\thbarweb - \thbarwb)(\thbarwb - t)]
&\leq E^\frac{1}{2}[(\thbarweb - \thbarwb)^2]E^\frac{1}{2}[(\thbarwb - t)^2] = o(m^{-1}).
\end{align*}
This concludes the proof of the theorem. \qedhere
\end{proof}